\documentclass[12pt]{iopart}
\expandafter\let\csname equation*\endcsname=\relax 
\expandafter\let\csname endequation*\endcsname=\relax 
\usepackage{amsmath,amssymb,amsthm,bbm}
\usepackage{xcolor}
\usepackage{graphicx}
\usepackage{xr}
\usepackage{comment}
\usepackage{cite}
\usepackage{tikz}
\usepackage{tabularx}
\usepackage{array}
\usepackage{dsfont}
\usepackage[unicode=true, breaklinks=false, pdfborder={0 0 1}, backref=false, colorlinks=true, linkcolor=blue, citecolor=blue]{hyperref}
\usetikzlibrary{matrix}

\newcommand\matthree[9]{%
  \begin{pmatrix}
  #1 & #2 & #3 \\ #4 & #5 & #6 \\ #7 & #8 & #9
  \end{pmatrix}%
}

\newcommand\T{\rule{0pt}{3.2ex}}       
\newcommand\B{\rule[-1.6ex]{0pt}{0pt}} 

\newcommand{\tmop}[1]{\ensuremath{\operatorname{#1}}}
\newcommand{\mathd}{\mathrm{d}}
\newcommand{\id}{\mathds{1}}
\newtheorem{prop}{Proposition}
\newtheorem{corollary}{Corollary}

\newcounter{dummy}
\begin{document}

\title[]{The interplay between local and non-local master equations: exact and approximated dynamics}

\author{Nina Megier, Andrea Smirne and Bassano Vacchini}

\address{Dipartimento di Fisica “Aldo Pontremoli”, Università degli Studi di Milano, via Celoria 16, 20133 Milan, Italy}
\address{Istituto Nazionale di Fisica Nucleare, Sezione di Milano, via Celoria 16, 20133 Milan, Italy}
\ead{nina.megier@mi.infn.it}

\vspace{10pt}
\begin{indented}
\item[]\today
\end{indented}

\begin{abstract}
Master equations are a useful tool to describe the evolution of open quantum systems.
In order to characterize the mathematical features and the physical origin of the dynamics, it is often useful
to consider different kinds of master equations for the same system. Here, we derive an exact connection
between the time-local and the integro-differential descriptions, focusing on the class of commutative dynamics. 
The use of the damping-basis formalism allows us to devise a general procedure to go from one master equation to the other and vice-versa, by working with functions of time and their Laplace transforms only.
We further analyze the Lindbladian form of the time-local and the integro-differential master equations, where we account for the appearance of different sets of Lindbladian operators.
In addition, we investigate a Redfield-like approximation, that transforms the exact integro-differential equation
into a time-local one by means of a coarse graining in time. 
Besides relating the structure of the resulting master equation to those associated with the exact dynamics, 
we study the effects of the approximation on Markovianity. In particular, we show that, against expectation, the coarse graining in
time can possibly introduce memory effects, leading to a violation of a divisibility property of the dynamics.    
\end{abstract}
\noindent{\it Keywords\/}: open quantum systems, master equations, memory kernels, quantum Markovianity


\section{Introduction}\label{sec:intr}
Any realistic physical system is unavoidably coupled to some external degrees of freedom and should then be
treated as an open system. This is especially relevant in the realm of quantum physics, which describes phenomena on small scales,
typically fragile under the interaction with the environment.
As a consequence, the theory of open quantum systems \cite{Breuer2002,Rivas2012} has a wide application area, e.g., in quantum chemistry \cite{TN_libero_mab214255846}, quantum information \cite{nielsenbook} and even biophysics \cite{TN_libero_mab21000245634,Huelga2013}. 

However, accounting for the interaction of the system of interest with the external degrees of freedom has its price. 
In particular, the equations of motion describing the evolution of open quantum systems
are either fully characterized, but derived under very restrictive assumptions, or in principle appropriate for a wide range of dynamical evolutions, but computationally demanding and at least partially unexplored. The former situation refers to the Gorini-Kossakowski-Sudarshan-Lindblad (GKSL) master equation \cite{Gorini1976,lindblad1976}, which is associated to quantum dynamical semigroups; the latter is based on more general time-local master equations or on integro-differential master equations, fixed by a memory kernel.
The focus of this paper is precisely on the connection between the latter types of description.

Both the time-local and the integro-differential descriptions of the open-system dynamics are highly relevant. The local one 
is better suited to access some of the properties of the evolution and to conduct calculations (e.g. numerically),
while the approach based on the use of memory kernels often gives a better insight of the physical processes underlying the dynamics. 
This is the case, for example, for quantum semi-Markov processes, where continuous in time quantum evolutions are randomly interrupted by jumps \cite{Budini2004,breuer2008,PhysRevE.79.041147,bassano_2012,PhysRevA.94.020103, PhysRevLett.117.230401}. The structure of the memory kernels for these processes clearly reflects this origin and additionally guarantees the property of complete positivity (CP) of the corresponding dynamical maps. 
For a given system it is in general of advantage to know both representations of its evolution, as each can extend one's knowledge and it is rarely possible to know a-priori which is the most suitable one.

The first goal of this paper is to shed some light on the connection between the descriptions given by the time-local
and the  integro-differential approach, both treated in an exact way. We use a damping-basis representation of the generator and memory kernel, which proves to be a powerful tool to characterize their
structural properties, as well as the connection between them \cite{PhysRevLett.104.070406}.
Focusing on commutative dynamics, i.e., such that the dynamical maps at different times commute \cite{Chruscinski2010,Chruscinski2014}, 
we define a 
general procedure to go directly from one form of the master equation to the other,
and viceversa.
In particular, we give an explicit link between the damping-basis description
of the time-local and the integro-differential master equations, which only relies on transformations
of functions of time and their Laplace transform. Furthermore, our analysis accounts for the differences and similarities
of the structural properties of the time-local and integro-differential master equations, when expressed in the canonical Lindbladian
form\cite{Gorini1976,lindblad1976,Breuer2002}. The possibly different sets of Lindblad operators of the two master equations
are traced back
to specific relations among the eigenvalues of the corresponding damping-basis
representations. In this way, as shown in the examples, we provide a common theoretical framework to various models considered in the physical literature.

The second part of the paper concerns the study of a possible simplified treatment of the integro-differential master equation,
which enforces a time-local structure, and can be seen as an approximated link between the two different descriptions of the open-system dynamics. 
Relying on the idea that 
the memory kernel is generally localized around the origin, with a width which can be interpreted 
as the memory time of the dynamics, one can define a general coarse-graining operation leading to a time-local master equation, along the lines of the seminal work of Redfield \cite{Redfield1957a}.
Here, we first show that the damping-basis description also enables us to connect this approximated time-local description and the two original, time-local and integro-differential, ones.  Furthermore, we study how moving to a Redfield-like master equation modifies the (non-)Markovian nature of the dynamics.

In recent years, the investigation of memory effects in the open quantum system dynamics, summarized under the term of “non-Markovianity", has attracted a great interest, among others, in connection with quantum thermodynamics \cite{Erez:2008ep,Strasberg2016,Pezzutto2016}, quantum-control theory \cite{Mukherjee2015} and quantum metrology \cite{Matsuzaki2011,PhysRevLett.109.233601,Smirne2016,Haase2018}. Till now, a wide range of non-equivalent characterizations of quantum non-Markovianity was introduced in the literature, see \cite{Rivas_2014, Breuer2016a, LiHallWiseman2017} for recent reviews. A lot of emphasis has been put on the divisibility properties of the dynamical maps as a possible signature of non-Markovian behavior. In particular, an open-system dynamics is called (C)P-divisible if the dynamical maps can be decomposed via (completely) positive propagators \cite{PhysRevLett.105.050403,Vacchini2011}.
CP-divisible dynamics correspond to so-called generalized GKSL master equations, which are time-local master equations with positive, but time dependent rates \cite{hallcresserandersson}; no equivalent general constraint is known for the memory-kernel master equations. 
The property of P-divisibility has a physical interpretation in terms of an information back-flow to the reduced system \cite{Bassano2015}, and appears to be related to a continuous-measurement interpretation \cite{Smirne2019}.
A characterization of P-divisible dynamics is only known for time-local generators \cite{Breuer2016a}, thanks to a result about positive semigroups \cite{Kossakowski1972a}, while the case of memory kernels remains undiscovered beyond some special cases.

From the structures of time-local and integro-differential equations one is tempted to conclude that the first description corresponds to “memoryless dynamics" and the latter introduces some memory effects, but neither statement is in general true. In particular, one can easily show that the time-local representation is always possible for dynamics given by invertible maps \cite{hallcresserandersson}, no matter how much non-Markovian the evolution is.
The time-local and the integro-differential descriptions seem to be somehow complementary: if one has a well behaving form, the other one is often singular, however both can describe non-Markovian dynamics. Here, with the help of the damping-basis approach, we will also investigate the relation between the (C)P-divisibility of the exact and Redfield-like dynamics. 
Despite the intuition that the coarse graining in time might wash out the memory effects due to the interaction
with the environment, we will show that imposing the Redfield-like master equation might actually result in going from a CP-divisible
dynamics into a non-CP-divisible one.

The rest of the paper is organized as follows. 
Section \ref{sec:structure} is concerned with the structure of the time-local generator and the memory kernel
associated to the same open-system dynamics.
After recalling the damping-basis formalism in
Section \ref{sec:dampingBasis}, we apply it to the class of commuting dynamics in Section \ref{sec:TCLvsNZ}. Here, the main result of the paper is given: Proposition \ref{prop:str}, where a direct connection between the damping-basis representation of
the time-local generator and the memory kernel is stated.
In Section \ref{sec:Lindblad} we further analyze the relation between the time-local
and the integro-differential descriptions using a Lindbladian form of the master equation. 
In Section \ref{sec:Redfield}, we investigate the
main features of the related Redfield-like master equation. We first work out the structure of the time-local generator obtained by approximating the exact memory kernel master equation in Section \ref{sec:Redfield1}. Finally, in Section \ref{sec:Redfield2} we show that this approximation, somehow counter-intuitively, does not always improves the divisibility properties of the time evolution. Section \ref{sec:summ} summarizes our findings.
 
\section{Structure of the master equations}\label{sec:structure}
The properties of the open quantum system alone are fixed by its density operator $\rho(t)$, also referred to as reduced state, at a generic time $t$. 
The dynamics of the open system is thus fully characterized by a family of dynamical maps $\{\Lambda_t\}_{t \geq 0}$, which map the initial reduced state $\rho(0)$ to the reduced state at later times, according to $\rho(t)=\Lambda_t[\rho(0)]$ ($\Lambda_0=\id$). 
In the time-local description of the dynamics, the family of dynamical maps satisfies the following equation 
\begin{equation}\label{eq:tcl}
\frac{\mathd}{\mathd t} \Lambda_t =  \mathcal{K}^{\tmop{TCL}}_{t} \Lambda_t,
\end{equation}
where the superscript $\tmop{TCL}$ stands for time-convolutionless \cite{Breuer2002};
instead, the integro-differential approach builds on the equation
\begin{equation}\label{eq:nz}
\frac{\mathd}{\mathd t} \Lambda_t = \int_0^t \mathd \tau \mathcal{K}^{\tmop{NZ}}_{t-\tau} \Lambda_\tau\equiv (\mathcal{K}^{\tmop{NZ}} * \Lambda)_t,
\end{equation}
where the symbol $*$ denotes the convolution in time, $\mathcal{K}^{\tmop{NZ}}_{t}$ is the memory kernel and $\tmop{NZ}$ stands for Nakajima-Zwanzig \cite{Breuer2002}.
The conditions on the time-local generator $\mathcal{K}^{\tmop{TCL}}_{t}$ and on the memory kernel $\mathcal{K}^{\tmop{NZ}}_{t}$ guarantying that $\Lambda_t$ is a proper dynamical map, i.e. completely positive\footnote{Complete positivity generalizes the property of positivity and takes into account the fact that the open quantum system can be entangled with some inaccessible degrees of freedom.} and trace preserving (CPTP), are in general unknown. Only in some limited cases definite statements in this respect have been obtained, see e.g. \cite{PhysRevLett.117.230401,Budini2004,Barnett2001,PhysRevA.70.010304,Maniscalco2007,Vacchini2013a,Vacchini2014a,Witt2017,PhysRevA.96.022129, Reimer2018,filippov2019phase,Vacchini2019,Nestmann2020a}.

\subsection{Damping bases}\label{sec:dampingBasis}
In this paper, we restrict to the case where the reduced system
is associated with a Hilbert space $\mathcal{H}$ of finite dimension $N$.
The set of linear operators on $\mathcal{H}$ is denoted as $\mathcal{B}(\mathcal{H})$ and it is an $N^2$ dimensional Hilbert space.
Hence, given a linear map $\Xi$ acting on $\mathcal{B}(\mathcal{H})$, often referred to as super-operator,
and a basis $\left\{\sigma_{\alpha}\right\}_{\alpha=1, \ldots, N^2}$
of $\mathcal{B}(\mathcal{H})$ orthonormal with respect to the Hilbert-Schmidt scalar product,
\begin{equation}\label{eq:orthobasis}
\langle \sigma_{\alpha} , \sigma_{\beta} \rangle = {\rm{Tr}}[ \, \sigma^{\dag}_{\alpha}\, \sigma_{\beta}] = \delta_{\alpha \beta},
\end{equation}
we can write the action of $\Xi$ on a generic
element $\omega \in \mathcal{B}(\mathcal{H})$ as \cite{Chruscinski2010,Smirne2010}
\begin{align}\label{eq:lexp}
\Xi(\omega) = \sum_{\alpha \beta =1}^{N^2} M^{\Xi}_{\alpha \beta} 
\rm{Tr}\left[\sigma^{\dag}_\beta \omega\right] \sigma_\alpha, && M^{\Xi}_{\alpha \beta} = \rm{Tr}\left[\sigma^{\dag}_{\alpha} \Xi(\sigma_{\beta}) \right].
\end{align}
The matrix $M^{\Xi}$ associated to the map $\Xi$ is in general not hermitian
and only allows for a spectral representation in Jordan form, leading to the
expression
\begin{equation}
  \label{eq:Jordan} \Xi (\omega) = \sum_{\alpha = 1}^M \left( \lambda_{\alpha}
  \sum_{\mu = 1}^{k_{\alpha}}  \rm{Tr} \left[ \varsigma^{\dag}_{\alpha, \mu} 
  \hspace{0.17em} \omega \right] \tau_{\alpha, \mu} + \sum_{\mu =
  1}^{k_{\alpha}}  \rm{Tr} \left[ \varsigma^{\dag}_{\alpha, \mu} 
  \hspace{0.17em} \omega \right] \tau_{\alpha,  \mu-1} \right) \mathrm{\qquad
  \text{with} \quad \omega \in \mathcal{B} (\mathcal{H}),}
\end{equation}
where $M \leqslant N^2$ is the number of eigenvalues each counted with its
geometric multiplicity, and $k_{\alpha}$ the algebraic multiplicity, reducing
to one if the matrix can be diagonalized. We define $\tau_{\alpha,0}=0$, and for $k_{\alpha} = 1$ we will consider the identifications
$\tau_{\alpha, 1} \equiv \tau_{\alpha}$ and $\varsigma_{\alpha, 1} \equiv
\varsigma_{\alpha}$. The families of $N^2$ operators $\{ \tau_{\alpha, \mu}
\}_{\alpha, \mu}$ and \ $\{ \varsigma_{\alpha, \mu} \}_{\alpha, \mu}$
appearing in Eq.~(\ref{eq:Jordan}) are related according to
\begin{equation}
  \label{eq:orthobasisGEN} \langle \varsigma_{\alpha, \mu}, \tau_{\beta, \nu}
  \rangle = \delta_{\alpha \beta} \delta_{\mu \nu} .
\end{equation}
The map $\Xi$ is said to be diagonalizable iff the corresponding matrix
$M^{\Xi}$ is. In this case, according to the previous identifications we can
write
\begin{equation}
  \label{eq:dampingdec} \Xi (\omega) = \sum_{\alpha = 1}^{N^2}
  \lambda_{\alpha} \mathrm{Tr \left[ \varsigma^{\dag}_{\alpha} 
  \hspace{0.17em} \omega \right]  \hspace{0.17em} \tau_{\alpha}  \qquad
  \text{with} \quad \omega \in \mathcal{B} (\mathcal{H}),}
\end{equation}
and the two bases of $\mathcal{B} (\mathcal{H})$, $\{ \tau_{\alpha} \}_{\alpha
= 1, \ldots, N^2}$ and $\{ \varsigma_{\alpha} \}_{\alpha = 1, \ldots, N^2}$,
which are not necessarily orthogonal bases, are sometimes referred to as
bi-orthogonal bases, obeying the orthogonality relation
\begin{equation}
  \label{eq:ortdamp} \langle \varsigma_{\alpha}, \tau_{\beta}
  \rangle = \delta_{\alpha \beta} .
\end{equation}
Within the context of open quantum systems, these bases were introduced in \cite{Briegel1993}
and named damping bases; in particular, there they were associated with a GKSL generator. In the following we will focus on the case of
diagonalizable maps, which is enough to cover all examples we will consider,
corresponding to relevant physical examples typically considered in the
literature, and in particular allows to consistently consider time dependence of the eigenvalues only.

The damping bases are also strictly connected to the relation
between the map $\Xi$ and its dual $\Xi'$, where the latter is defined by
\begin{equation}\label{eq:dual}
\langle \omega , \Xi(\rho)\rangle =\langle \Xi'(\omega), \rho \rangle \quad \forall \omega, \rho \in 
\mathcal{B}(\mathcal{H}).
\end{equation}
Using the damping bases, one finds
\begin{equation}\label{eq:dualmapdiag}
\Xi'(\omega) = \sum_{\alpha=1}^{N^2} \lambda^*_{\alpha } \rm{Tr}\left[\tau^{\dag}_{\alpha}\, \omega\right] \varsigma_{\alpha} \qquad \omega \in \mathcal{B}(\mathcal{H}),
\end{equation}
where $c^*$ is the complex conjugate of $c$.
From Eqs.~(\ref{eq:dampingdec}), (\ref{eq:ortdamp}) and (\ref{eq:dualmapdiag}) one can then see that the operators 
$\left\{\tau_{\alpha}\right\}_{\alpha = 1, \ldots, N^2}$ and $\{\varsigma_{\alpha}\}_{\alpha = 1, \ldots, N^2}$ 
are the eigenvectors, respectively, of the linear map $\Xi$ and of its dual $\Xi'$ with respect 
to complex conjugates eigenvalues, i.e.
\begin{equation}
\Xi (\tau_{\alpha}) = \lambda_{\alpha}\, \tau_{\alpha}, 
\qquad \,\, \Xi' (\varsigma_{\alpha}) = \lambda^*_{\alpha}\, \varsigma_{\alpha} \quad \alpha=1,\ldots,N^2. \label{eq:eigen}  
\end{equation}
Indeed, if $\Xi$ is a normal operator, i.e., $[\Xi, \Xi']=0$,
then the damping bases both coincide with a single orthonormal basis; if $\Xi$ is a Hermitian map, in addition
the eigenvalues are real, $\lambda_\alpha=\lambda^*_\alpha$.

Let us now move to one-parameter families of maps, $\left\{\Xi_t\right\}_{t\geq 0}$,
which are used to describe the dynamics of open quantum systems.
In general, both the coefficients and the operators in Eq.~(\ref{eq:dampingdec})
will depend on time. However, it can well happen that the time dependence is enclosed in the eigenvalues
only, while the corresponding damping bases are time-independent, i.e., that one has
\begin{equation}\label{eq:dampingtime}  
\Xi_t (\omega) = \sum_{\alpha=1}^{N^2} \lambda_{\alpha}(t) 
\rm{Tr}\left[\varsigma^{\dag}_{\alpha}\, \omega\right]\, \tau_{\alpha} \qquad \omega \in \mathcal{B}(\mathcal{H}). 
\end{equation}
Eq.~(\ref{eq:dampingtime}) implies that the maps at different times commute,
\begin{equation}\label{eq:comm}
\left[\Xi_t , \Xi_s \right] = 0;
\end{equation}
the converse implication holds, if we further assume that the maps $\Xi_s$ and $\Xi_t$ are diagonalizable.
In particular, one can consider a one-parameter family of CPTP dynamical maps
$\left\{\Lambda_t\right\}_{t\geq0}$ which commute at different times (a situation which has been thoroughly investigated in \cite{Chruscinski2010,Chruscinski2014}),
i.e., such that 
\begin{equation}\label{eq:comm2}
    \left[\Lambda_t, \Lambda_s \right] = 0 \qquad \forall t, s\geq0.
\end{equation}
If the family $\left\{\Lambda_t\right\}_{t\geq0}$ satisfies a time-local master equation of the form Eq.~\eqref{eq:tcl}
then the dynamical maps satisfy Eq.~(\ref{eq:comm2}) if and only if the time-local generator
satisfies the analogous commutation relation
\begin{equation}\label{eq:comtcl}
    \left[\mathcal{K}^{\tmop{TCL}}_t, \mathcal{K}^{\tmop{TCL}}_s \right] = 0 \qquad \forall t, s\geq0.
\end{equation}
If we further assume diagonalizability of the dynamical maps and the generator, they will 
share the same time-independent damping bases, see Eq.~(\ref{eq:dampingtime}), 
according to:
\begin{eqnarray}
  \mathcal{K}^{\tmop{TCL}}_t&=& \sum_{\alpha=1}^{N^2}  m^{\tmop{TCL}}_{\alpha}(t)
\rm{Tr}\left[\varsigma^{\dag}_{\alpha}\, \omega\right]\, \tau_{\alpha}, \label{eq:tcldiag}\\
  \Lambda_t &=& \sum_{\alpha=1}^{N^2}m_{\alpha}(t)
                \rm{Tr}\left[\varsigma^{\dag}_{\alpha}\, \omega\right]\, \tau_{\alpha}, \label{eq:mapdiag}
\end{eqnarray}
where
\begin{equation}\label{eq:map-tcl}
    m_{\alpha}(t)= e^{\int_0^t\mathd \tau 
m^{\tmop{TCL}}_{\alpha}(\tau)} .
\end{equation}
Let us stress that the dynamics satisfying Eq.~(\ref{eq:comm2}) include, but are not restricted to,
the case where Eq.~(\ref{eq:tcl}) holds with $\mathcal{K}^{\tmop{TCL}}_t = \gamma(t) \mathcal{L}$, which has been widely studied in the literature \cite{Barnett2001,PhysRevA.70.010304,PhysRevA.71.020101,PhysRevA.73.012111,PhysRevA.81.062120};
as relevant examples, let us mention pure dephasing, depolarization, spontaneous emission, 
two-level system in the presence of an infinite-temperature bosonic bath, photonic losses,
as well as the composition of local evolutions made up of these dynamics.

\subsection{Novel connection between time-local and integro-differential representation}\label{sec:TCLvsNZ}

Given a diagonalizable commutative dynamics, its damping-basis decomposition sets 
a representation of both the dynamical map and the time-local generator in which the time dependence
is fully enclosed in the eigenvalues -- i.e., in complex functions of time -- but not in the operatorial
structure, see Eqs.~(\ref{eq:tcldiag}) and (\ref{eq:mapdiag}). 
Here, we show that this feature is also shared by the memory kernel, which will allow us to derive some novel connections
between the time-local and the integro-differential master equations associated to a given dynamics.

\begin{prop}\label{prop:str} 
Consider a family of dynamical maps $\left\{\Lambda_t\right\}_{t\geq0}$ with time-local generator $\mathcal{K}^{\tmop{TCL}}_t $ and  memory kernel $\mathcal{K}^{\tmop{NZ}}_t $.
Moreover, for any couple of operators $\tau_{\alpha}$ and $\varsigma_{\alpha}$
let $\mathcal{M}_\alpha$ be the linear map acting on $\mathcal{B}(\mathcal{H})$ defined as
\begin{equation}
\mathcal{M}_\alpha(\omega)=\rm{Tr}\left[\varsigma^{\dag}_{\alpha}\, \omega\right]\, \tau_{\alpha} 
\quad \omega \in \mathcal{B}(\mathcal{H}). 
\end{equation}
The following propositions are equivalent:
\begin{itemize}
    \item[i)] the time-local generator $\mathcal{K}^{\tmop{TCL}}_t $ has the damping-basis diagonalization
    \begin{eqnarray}
\mathcal{K}^{\tmop{TCL}}_t &=& \sum_{\alpha=1}^{N^2} m^{\tmop{TCL}}_{\alpha}(t) \mathcal{M}_{\alpha}; \label{eq:ktcl}
\end{eqnarray}
\item[ii)] the memory kernel $\mathcal{K}^{\tmop{NZ}}_t $ has the damping-basis diagonalization
\begin{eqnarray}
\mathcal{K}^{\tmop{NZ}}_t &=& \sum_{\alpha=1}^{N^2} m^{\tmop{NZ}}_{\alpha}(t) \mathcal{M}_{\alpha}; \label{eq:knz}
\end{eqnarray}
\end{itemize}
moreover, the corresponding eigenvalues are related by
\begin{eqnarray}
  m^{\tmop{NZ}}_{\alpha}(t) &=& \mathfrak{I}\left(\frac{u \widetilde{G_{\alpha}}(u)}{1+\widetilde{G_{\alpha}}(u)}\right)(t),
                                \label{eq:mnz}
                                \\
                                m^{\tmop{TCL}}_{\alpha}(t) &=& \frac{G_{\alpha} (t)}{1 + \int_0^t \mathd \tau G_{\alpha} (\tau)},
                                \label{eq:tcltcl}\\
G_{\alpha}(t) &=&\mathfrak{I}\left(\frac{\widetilde{m^{\tmop{NZ}}_{\alpha}}(u)}{u-\widetilde{m^{\tmop{NZ}}_{\alpha}}(u)}\right)(t),  \label{eq:mtcl}
\end{eqnarray}
where we introduced the function
\begin{eqnarray}\label{eq:ggammat}
G_{\alpha}(t) &=& \frac{\mathd}{\mathd t} e^{\int_0^t \mathd \tau m^{\tmop{TCL}}_{\alpha}(\tau)};
\end{eqnarray}
$\widetilde{f_t}(u) \equiv \widetilde{f}_u$ is the Laplace transform of $f_t$
and $\mathfrak{I}\left(g(u)\right)(t)$
the inverse Laplace transform of $g(u)$.
The following relationships further hold
\end{prop}
\begin{proof} The proof is given in \ref{app:proofp1}.
\end{proof}

Proposition \ref{prop:str} is the central result of this paper. According to it, 
the time-local generator and the memory kernel have the same 
time-independent damping bases, while the time dependence appears in the eigenvalues only: to go from one
master equation to the other, only functions of time (and their Laplace transforms) are involved,
while the operatorial structure is unchanged; see the left part of Fig.~\ref{fig:prop:unique}.

With reference to the expression Eq.~(\ref{eq:mapdiag})  of the associated time evolution as shown in the proof
Proposition \ref{prop:str} comes along with the relations
\begin{eqnarray}
  m_{\alpha} (t) & = & e^{\int_0^t \mathd \tau m^{\tmop{TCL}}_{\alpha}
  (\tau)} ,  \label{eq:tcl-map1}\\
  m^{\tmop{TCL}}_{\alpha} (t) & = & \frac{1}{m_{\alpha} (t)}
  \frac{\mathd}{\tmop{dt}} m_{\alpha} (t) , \label{eq:tcl-map2}
\end{eqnarray}
as well as 
\begin{eqnarray}
  \widetilde{m^{}_{\alpha}} (u) & = & \frac{1}{u -
  \widetilde{m_{\alpha}^{\tmop{NZ}}} (u)} , \label{eq:nz-map1}\\
  \widetilde{m_{\alpha}^{\tmop{NZ}}} (u) & = & \frac{u
  \widetilde{m^{}_{\alpha}} (u) - 1}{\widetilde{m^{}_{\alpha}} (u)} .
  \label{eq:nz-map2}
\end{eqnarray}
 
\begin{figure}[!tbp]
   \includegraphics[width=\textwidth]{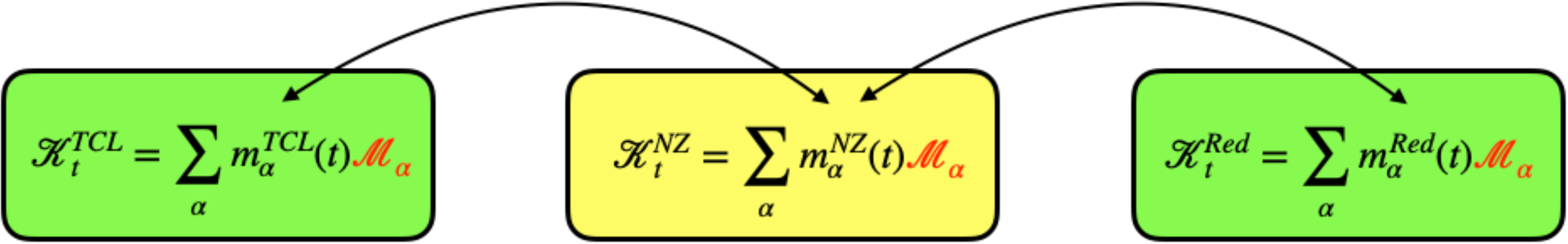}
  \caption{According to Proposition \ref{prop:str} and \ref{prop:red}, 
  the time-local generator, the memory kernel and the Redfield-like generator of commutative dynamics have the same 
time-independent damping bases, while the time dependence appears in the eigenvalues only: to go from one
to the other, only functions of time (and their Laplace transforms) are involved,
while the operatorial structure is unchanged.}
  \label{fig:prop:unique}
\end{figure}
As a consequence, Proposition \ref{prop:str} is particularly useful when 
we want to compare the structures of the time-local generator and the memory kernel.
As a general relevant example, consider the following corollary, which directly follows from the proposition above.
\begin{corollary}\label{cor:str}
Under the same assumptions of Proposition \ref{prop:str}, consider a diagonalizable GKSL generator $\mathcal{L}$ with only one-nonzero
eigenvalue $\ell$, possibly degenerate with degeneracy $d$, i.e.,
\begin{equation}\label{eq:dampingldeg}  
\mathcal{L} = \ell \sum_{\alpha=1}^d \mathcal{M}_{\alpha};
\end{equation}
then the following identities are equivalent
\begin{eqnarray}
\mathcal{K}^{\tmop{TCL}}_t &=& \gamma(t) \mathcal{L};\label{eq:tcll}\\
\mathcal{K}^{\tmop{NZ}}_t &=& \frac{m^{\tmop{NZ}}(t)}{\ell}\mathcal{L}; \label{eq:nzl} 
\end{eqnarray}
where $m^{\tmop{NZ}}(t)$ and $G (t)$ are related as in Eqs.~\eqref{eq:mnz} and \eqref{eq:mtcl}
where now
\begin{equation}
   G(t) = \frac{\mathd}{\mathd t} e^{\ell \int_0^t \mathd \tau \gamma(\tau)}. \label{eq:mnz2}  
\end{equation}
\end{corollary}
\begin{proof}
Simply note that $m^{\tmop{TCL}}(t)= \gamma(t) \ell$ and $\sum_{\alpha=1}^d \mathcal{M}_{\alpha}=\mathcal{L}/\ell$,
see Eq.~(\ref{eq:dampingldeg});
then apply Proposition \ref{prop:str}.
\end{proof}

\noindent In other terms, if $\mathcal{L}$ has only one non-zero eigenvalue, the super-operatorial part of the time-local generator
and the memory kernel are exactly the same;
the difference between the master equations is enclosed in one overall time-dependent factor.
A physically relevant example of this situation will be given below, see Example \hyperref[ex:3]{3}.

In Proposition \ref{prop:str} we assumed commutativity of the dynamics. As
mentioned at the end of Sec.\ref{sec:dampingBasis}, and as we will show
explicitly by means of example, such maps account for several dynamics of
interest for open quantum systems. On the other hand, it is
indeed natural to ask what happens if we relax this assumption. While a full-
fledged extension of the result goes beyond the scope of this work, we show
here that Proposition \ref{prop:str} is indeed also valid for more general
dynamics in modified form, in which a time-dependent damping basis, different for time-local generator and memory kernel, has to be considered.

A relevant class of non-commutative evolutions for which this is the case 
is provided by the phase-covariant qubit dynamics.
Under such an evolution, the Bloch ball shrinks into a possibly rotated and
shifted ellipsoid, such that the overall transformation commutes with the
rotation about a fixed axis. Recently, phase-covariant dynamics have attracted
considerable attention
{\cite{Vacchini2010,Smirne2016,Lostaglio2017,Haase2018,Teittinen2018,filippov2019phase}},
both because of their clear mathematical meaning and structure, and physical
relevance e.g. for metrological tasks {\cite{Haase2018}}. The expression of
the time-local and integro-differential master equation for phase-covariant
dynamics has been worked out in {\cite{filippov2019phase}}, and the
result indeed complies with Eqs. {\eqref{eq:mnz}}-{\eqref{eq:ggammat}}. One
can observe that in this case the validity of these relations follows from
\begin{align}
  \label{conditionPC} m^{\tmop{TCL}}_{\alpha} (t) m^{\tmop{TCL}}_{\beta} (t')
  \mathcal{M}^{\tmop{TCL}}_{\alpha} (t) \mathcal{M}^{\tmop{TCL}}_{\beta} (t')
  = \delta_{\alpha \beta} m^{\tmop{TCL}}_{\alpha} (t) m^{\tmop{TCL}}_{\alpha}
  (t') \mathcal{M}^{\tmop{TCL}}_{\alpha} (t'),
\end{align}
which is indeed a strictly weaker requirement than commutativity
whenever at least one of the eigenvalues
$m^{\tmop{TCL}}_{\alpha} (t)$ is equal to zero, e.g, if the dynamics has at least one steady state
{\cite{Rivas2012}}. The linear
map $\mathcal{M}^{\tmop{TCL}}_{\alpha} (t)$ corresponds to the time-dependent
damping basis of the time local generator: $\mathcal{K}^{\tmop{TCL}}_t =
\sum_{\alpha} m^{\tmop{TCL}}_{\alpha} (t) \mathcal{M}^{\tmop{TCL}}_{\alpha}
(t)$. Eq.~{\eqref{conditionPC}} warrants that the
dynamical map can be written as $\Lambda_t = \sum_{\alpha} m_{\alpha}
(t) \mathcal{M}_{\alpha} (t)$ and the memory kernel as
$\mathcal{K}^{\tmop{NZ}}_t = \sum_{\alpha} m^{\tmop{NZ}}_{\alpha} (t)
\mathcal{M}^{\tmop{NZ}}_{\alpha} (t)$, where the connection between
$m_{\alpha} (t), m^{\tmop{TCL}}_{\alpha} (t)$ and
$m^{\tmop{NZ}}_{\alpha} (t)$ is still given by Eqs.{\eqref{eq:tcldiag}}-{\eqref{eq:mapdiag}} and Proposition \ref{prop:str}, though of course now
the respective damping bases will not coincide.

\subsection{Lindbladian form}\label{sec:Lindblad}

The time-local generator and the memory kernel in Corollary \ref{cor:str} 
are directly proportional to a GKSL generator of quantum dynamical semigroups.
Here we show that, starting from the damping-basis decomposition in Eqs.~\eqref{eq:ktcl} and \eqref{eq:knz}, it is always
possible to write the time-local generator and the memory kernel in a way which is directly related to the GKSL generator; 
we will refer to such form as Lindbladian.
To do so, we essentially apply the general prescription given by Lemma 2.3 in \cite{Gorini1976}
to the situation of interest for us.
Besides providing us with a canonical reference structure which eases the comparison between super-operators, as we show
by different examples, the 
Lindbladian form allows us to infer the (C)P and the (C)P-divisibility of the dynamics
in a more direct way.

Any linear map acting on $\mathcal{B}(\mathcal{H})$ can be represented in several ways.
A relevant example is given by the matrix representation in Eq.~\eqref{eq:lexp},
which is at the basis of the damping-basis decomposition, see Eq.~\eqref{eq:dampingdec}. 
Alternatively, given an orthonormal basis $\left\{\sigma_{\alpha}\right\}_{\alpha=1, \ldots, N^2}$
of $\mathcal{B}(\mathcal{H})$, see Eq.~\eqref{eq:orthobasis}, with $\sigma_{N^2}=\mathds{1}/\sqrt{N}$,
any linear map $\Xi$ acting on $\mathcal{B}(\mathcal{H})$ can be uniquely written
as 
\begin{equation}\label{eq:decomp2}
    \Xi(\omega) = \sum_{\alpha,\beta=1}^{N^2}c_{\alpha, \beta} \,\, \sigma_{\alpha}\omega \sigma^{\dag}_{\beta}
    \qquad \mbox{with} \quad \omega \in \mathcal{B}(\mathcal{H}).
\end{equation}
Note that such a representation is strictly related to the CP of the map $\Xi$: in fact, the matrix of coefficients
with elements $c_{\alpha \beta}$ is positive semidefinite if and only if $\Xi$ is CP, in which case the decomposition
in Eq.~\eqref{eq:decomp2} directly leads to the Kraus decomposition of CP maps.
Most importantly for us, the representation in Eq.~\eqref{eq:decomp2} gives a general characterization also of the time-local
and the integro-differential master equations associated with open-system dynamics. In fact,
let us consider a time-local generator or a memory kernel $\mathcal{K}_t$.    
The dynamics $\{\Lambda_t\}_{t \geq 0}$ resulting from Eqs.~\eqref{eq:tcl} and \eqref{eq:nz} is trace and Hermiticity preserving 
(where the latter means that any Hermitian operator $\omega = \omega^{\dag}$ is mapped at any time $t$
into an Hermitian operator $\Lambda_t(\omega)=(\Lambda_t(\omega))^{\dag}$) if and only if the corresponding
time-local generator and memory kernel satisfy\footnote{The conditions in Eq.~(\ref{eq:cond}) have to be satisfied
by any time-local generator, if we assume that $\Lambda_t^{-1}$ exists: 
this can be checked by applying Eq.~(\ref{eq:tcl}) to a generic initial
state $\rho_0$ and then taking the trace, for the first condition, and the Hermitian conjugate, for the second one, on both sides.
Analogously, they have to satisfied by any memory kernel, if we assume
that $\tilde{\Lambda}_u^{-1}$ exists, as can be checked by
applying Eq.~(\ref{eq:nz}) to
a generic initial state and taking the Laplace transform.\\
}
\begin{eqnarray}
    \rm{Tr}\left[\mathcal{K}_t (\omega)\right] &=& 0 \qquad \forall \omega\in \mathcal{B}(\mathcal{H}), \nonumber\\
    \left(\mathcal{K}_t(\omega)\right)^{\dag} &=& \mathcal{K}_t (\omega^{\dag}) \qquad \forall \omega\in \mathcal{B}(\mathcal{H}). \label{eq:cond}
\end{eqnarray}
Restricting for the sake of convenience to diagonalizable super-operators, in the form
\begin{equation}
\mathcal{K}_t=\sum_{\alpha=1}^{N^2} m_{\alpha}(t) \mathcal{M}_{\alpha},    
\end{equation}
Lemma 2.3 of \cite{Gorini1976} tells us that the two conditions in Eq.~\eqref{eq:cond}
hold if and only if
\begin{equation}\label{eq:gorini23}
\mathcal{K}_t (\omega) = -i \left[ H(t) , \omega \right] + \sum^{N^2-1}_{\alpha, \beta =1} \kappa_{\alpha \beta}(t) \left(\sigma_{\alpha} \omega \sigma^{\dag}_{\beta}
-\frac{1}{2} \left\{\sigma^{\dag}_{\beta} \sigma_{\alpha} , \omega \right\} \right),
\end{equation}
where the coefficients $\kappa_{\alpha \beta}(t)=\kappa_{\beta\alpha}(t)^*$ are given by
\begin{equation}
\kappa_{\alpha \beta}(t) = \sum_{\gamma=1}^{N^2}\sum_{\chi=1}^{N^2} m_{\gamma}(t)  \label{eq:coeffl}
\rm{Tr}\left[\sigma_{\beta}\sigma_{\chi}^{\dag} \sigma_{\alpha}^{\dag}\mathcal{M}_{\gamma}(\sigma_{\chi})\right],
\end{equation}
while the Hamiltonian $H(t)=H^{\dag}(t)$ is given by
\begin{equation}\label{eq:htf}
H(t) = \frac{1}{2 i} (\sigma^{\dag}(t) - \sigma(t)), \qquad
\sigma(t) = \frac{1}{\sqrt{N}} \sum^{N^2-1}_{\alpha=1} \kappa_{\alpha N^2}(t) \sigma_{\alpha}.
\end{equation}
In addition, since the matrix with elements $\kappa_{\alpha \beta}(t)$ is Hermitian, it can be diagonalized 
by a unitary matrix $V(t)$ with elements $ V_{\alpha \beta}(t)$,
so that Eq.~\eqref{eq:gorini23} can be rewritten as
\begin{equation}\label{eq:gorinidiag}
\mathcal{K}_t (\omega) = -i \left[ H(t) , \omega \right] + \sum^{N^2-1}_{\alpha=1} r_{\alpha}(t) \left(L_{\alpha}(t) \omega L^{\dag}_{\alpha}(t)
-\frac{1}{2} \left\{L^{\dag}_{\alpha}(t) L_{\alpha}(t) , \omega \right\} \right),
\end{equation}
with
\begin{eqnarray}
r_{\alpha}(t) &=& \sum_{\gamma, \gamma'=1}^{N^2-1} 
V^{\dag}_{\alpha \gamma}(t) \kappa_{\gamma \gamma'}(t)V_{\gamma' \alpha}(t) \qquad r_{\alpha}(t) \in \mathbb{R}                 \nonumber\\
L_{\alpha}(t) &=& \sum_{\beta=1}^{N^2-1} V_{\beta \alpha}(t)\, \sigma_{\beta}. \label{eq:diagop}
\end{eqnarray}
Eqs.~(\ref{eq:coeffl})-(\ref{eq:diagop})
provide us with the wanted recipe to get the Lindbladian form,
starting from the damping-basis representation, Eqs.~(\ref{eq:ktcl}) and (\ref{eq:knz}).
The main difficulty is that the different non-zero eigenvalues $m_{\alpha}(t)$
will “mix" in a non-trivial way, so that there is not a direct connection between them and the 
coefficients $r_{\alpha}(t)$ in the Lindbladian structure.
As a first consequence, one looses the correspondence between the super-operatorial structures of, respectively, the time-local generator
and the memory kernel,
which is guaranteed by Corollary \ref{cor:str} in the case of one non-zero eigenvalue; this is shown by the examples below.\\

\noindent{\textbf{Example 1}} \label{ex:1}
Consider the time-local generator 
\begin{equation}\label{eq:tcl0}
\mathcal{K}^{\tmop{TCL}}_t(\omega) =
\gamma_-(t)\left(\sigma_- \omega \sigma_+ - \frac{1}{2}\left\{\sigma_+ \sigma_-, \omega\right\}\right),
\end{equation}
which describes, for example, the reduced dynamics of a two-level system interacting with a zero-temperature
bosonic bath via a Jaynes-Cummings interaction term \cite{Vacchini2010,Smirne2010}, neglecting for the sake of simplicity the free Hamiltonian term.
This is a special case of the generator treated (for constant coefficients) in \cite{Briegel1993}.
The dual generator is 
\begin{equation}\label{eq:tcl0d}
\left(\mathcal{K}^{\tmop{TCL}}_t\right)'(\omega) =
\gamma_-(t)\left(\sigma_+ \omega \sigma_- - \frac{1}{2}\left\{\sigma_+ \sigma_-, \omega\right\}\right),
\end{equation}
and the resulting eigenvalues and damping bases are [see Eq.~(\ref{eq:eigen})] \footnote{Note that in \cite{Briegel1993} they used a different duality relation, without the Hermitian conjugate, so
that there is a different dual damping basis with respect to ours.}
\begin{eqnarray}
\left\{m^{\tmop{TCL}}_\alpha(t)\right\}_{\alpha=1,\ldots, 4} &=& \left\{0,-\gamma_-(t),-\frac{1}{2}\gamma_-(t),-\frac{1}{2}\gamma_-(t)\right\}\label{eq:valueex2}\\
\left\{\tau_\alpha\right\}_{\alpha=1,\ldots, 4} &=& \left\{\frac{\mathds{1}-\sigma_z}{2},\sigma_z, \sigma_+, \sigma_-\right\} \label{eq:vecex2}\\
\left\{\varsigma_\alpha\right\}_{\alpha=1,\ldots, 4} &=& \left\{\mathds{1},\frac{\mathds{1}+\sigma_z}{2},\sigma_+, \sigma_-\right\}. \label{eq:vecdex2}
\end{eqnarray}
Indeed, the relations in Eq.~(\ref{eq:ortdamp}) are satisfied; moreover, we note that we have now two eigenvalues different from zero
(one two-fold degenerate)
and the damping bases are not made of self-adjoint operators.

The time dependence is enclosed in the eigenvalues only, so that we are in the case of commuting dynamics treated in the previous sections.
In particular, by applying Proposition \ref{prop:str},
we find that the integro-differential generator $\mathcal{K}^{\tmop{NZ}}_t$
has the same damping bases, Eqs.~(\ref{eq:vecex2}) and (\ref{eq:vecdex2}), with eigenvalues given by Eqs.~(\ref{eq:mnz}) and (\ref{eq:ggammat})
with respect to the functions in Eq.~(\ref{eq:valueex2}).
In other terms, both the time-local generator and the memory kernel can be written in the form [see Eq.~(\ref{eq:dampingdec})]
\begin{equation}
    \mathcal{K}^X_t(\omega) = \frac{m^X_2(t)}{2}{\rm{Tr}}\left[(\mathds{1}+\sigma_z)\omega\right] \sigma_z
    +m^X_3(t)\left({\rm{Tr}}\left[\sigma_-\omega\right] \sigma_++{\rm{Tr}}\left[\sigma_+\omega\right] \sigma_-\right),
    \quad X = \tmop{TCL}, \tmop{NZ} \label{eq:auxx}
\end{equation}
with $m^X_i(t)$ the corresponding eigenvalues.
But, by using Eqs.~(\ref{eq:coeffl})-(\ref{eq:diagop}), one can see that the generator as in Eq.~(\ref{eq:auxx}) corresponds
to a Lindbladian form
\begin{equation}
  \mathcal{K}^X_t(\omega) =
-m^X_2(t)\left(\sigma_- \omega \sigma_+ - \frac{1}{2}\left\{\sigma_+ \sigma_-, \omega\right\}\right)  
+\frac{1}{2}(m^X_2(t)- 2m^X_3(t))\left(\sigma_z \omega \sigma_z - \omega\right).\label{eq:ex2fin}
\end{equation}
Crucially, while for $\mathcal{K}^{\tmop{TCL}}_t$ one has $m^{\tmop{TCL}}_2(t)=2m^{\tmop{TCL}}_3(t)$ [see Eq.~(\ref{eq:valueex2})],
so that the pure dephasing term cancels out, this is generally not the case for $\mathcal{K}^{\tmop{NZ}}_t$,
which will then present a pure-dephasing term, in addition to the spontaneous-emission term;
an explicit example of this is given in \cite{Vacchini2010,Smirne2010}. This implies in particular that at variance with the case of the standard GKSL generator, a direct physical interpretation of the operatorial contribution is not available.
\\
As anticipated, even though the time-local master equation we started from, Eq.~(\ref{eq:tcl0}),
is in the form given by Eq.~(\ref{eq:tcll}), now the corresponding GKSL generator has more than one eigenvalue
different from 0. As a consequence, the transformation to the memory kernel no longer preserves the super-operatorial
part of the Lindbladian structure [compare with Corollary \ref{cor:str}], but rather generates one more term.\\

\noindent {\textbf{Example 2}} \label{ex:2}
A somehow opposite example is obtained by starting with a memory kernel of the form 
\begin{equation}
   \mathcal{K}^{\tmop{NZ}}_t(\omega) = k(t) \left(\sigma_- \omega \sigma_+ +\sigma_+ \omega \sigma_-
    -\omega\right),
\end{equation}
with $k(t)$ such that the described dynamics is CPTP, as e.g. in \cite{Vacchini2011}.
First, we note that such a generator can be written in a ``manifest'' Lindbladian structure as
\begin{equation}\label{eq:exnz2}
   \mathcal{K}^{\tmop{NZ}}_t(\omega) = k(t) \left(\sigma_- \omega \sigma_+ -\frac{1}{2}(\sigma_+ \sigma_- \omega + \omega \sigma_+ \sigma_- ) 
   +\sigma_+ \omega \sigma_- -\frac{1}{2}(\sigma_- \sigma_+ \omega + \omega \sigma_- \sigma_+ )
    \right),
\end{equation}
so that this time we start from a non-local generator in the form 
$\mathcal{K}^{\tmop{NZ}}_t = k(t) \mathcal{L}$ [compare with Eq.~(\ref{eq:tcll})].
Such a generator is self-adjoint, $\left(\mathcal{K}^{\tmop{NZ}}_t\right)'=\mathcal{K}^{\tmop{NZ}}_t$, so that the
eigenvalues are real, and the damping bases coincide, yielding an orthonormal basis:
\begin{eqnarray}
\left\{m^{\tmop{NZ}}_\alpha(t)\right\}_{\alpha=1,\ldots, 4} &=& \left\{0,-k(t),-k(t),-2 k(t)\right\}\label{eq:valueex3}\\
\left\{\tau_\alpha\right\}_{\alpha=1,\ldots, 4} &=& \left\{\frac{\mathds{1}}{\sqrt{2}},\frac{\sigma_x}{\sqrt{2}}, \frac{\sigma_y}{\sqrt{2}}, \frac{\sigma_z}{\sqrt{2}}\right\}. \label{eq:vecex3}
\end{eqnarray}
Also in this case, since we have two eigenvalues different from 0 we cannot apply Corollary \ref{cor:str},
which would guarantee that the time-local equation would have the same Lindbladian structure.
Instead, we can apply Proposition \ref{prop:str}, so that both the time-local generator and the memory kernel
have the form
\begin{equation}
  \mathcal{K}^X_t(\omega) =
    m^X_2(t)\left({\rm{Tr}}\left[\sigma_x \omega\right] \sigma_x + {\rm{Tr}}\left[\sigma_y \omega\right] \sigma_y\right)
    +m^X_4(t) {\rm{Tr}}\left[\sigma_z\omega\right] \sigma_z,  \qquad X = \tmop{TCL}, \tmop{NZ}\label{eq:ex3aux}
\end{equation}
where for the memory kernel the eigenvalues are given by Eq.~(\ref{eq:valueex3}),
while for the time-local generator they are obtained from the latter via Eqs.~(\ref{eq:mtcl}) and (\ref{eq:ggammat}).
Using Eqs.~(\ref{eq:coeffl})-(\ref{eq:diagop}), Eq.~(\ref{eq:ex3aux}) can be written in Lindbladian form as
\begin{eqnarray}
  \hspace*{-2,5cm}  \mathcal{K}^X_t(\omega) = -\frac{m^X_4(t)}{2} \left(\sigma_- \omega \sigma_+ -\frac{1}{2}(\sigma_+ \sigma_- \omega + \omega \sigma_+ \sigma_- ) 
   +\sigma_+ \omega \sigma_- -\frac{1}{2}(\sigma_- \sigma_+ \omega + \omega \sigma_- \sigma_+ )
    \right) \nonumber
    \\\hspace*{-0,8cm} 
    +\frac{1}{2}(m^X_4(t)-2m^X_2(t))\left(\sigma_z \omega \sigma_z - \omega\right).
\end{eqnarray}
For the memory kernel $m^{\tmop{NZ}}_4(t)=2m^{\tmop{NZ}}_2(t)$, so that the pure-dephasing term cancels out, while
this will not generally be the case for the time-local generator. An example is given in \cite{Vacchini2011}.
Once again, the multiple eigenvalues in the damping-basis decomposition generate further terms, this time when going from the integro-differential
to the time-local master equation.
Of course, the situation is symmetrical, so that we could obtain further examples by simply inverting the starting points in the examples above;
the only difference when going, respectively, from the time-local to the integro-differential master equation or viceversa is
the connection between the eigenvalues of the damping-basis decomposition, i.e., whether one should use Eq.~(\ref{eq:mnz}) or Eq.~(\ref{eq:mtcl}).

We conclude that Propositions \ref{prop:str} and Eqs.~(\ref{eq:coeffl})-(\ref{eq:diagop}) yield a systematic procedure
to obtain the integro-differential master equation from the time-local one and viceversa, whenever we are able to diagonalize
them and the resulting damping bases are time-independent. An example of application for a higher dimensional system is given in \ref{app:proofhigherdim}.

\section{Redfield-like master equation}\label{sec:Redfield}

Until now, we compared the time-local and integro-differential master equations in an exact way, i.e., without introducing
any approximation. On the other hand, as recalled in the Introduction, it is useful to consider situations
in which an integro-differential master equation is transformed into a time-local one by means of some approximations.
Here, we focus on a master equation which is obtained via a coarse graining in time
analogous to the one introduced by Redfield in \cite{Redfield1957a}.

More precisely, if $\tau_R$ is the relaxation time of the open-system
dynamics and $\mathcal{K}^{\tmop{NZ}}_{t}$ is appreciably different from zero only on a time scale much shorter than $\tau_R$, 
one might approximate the dynamical maps 
$\Lambda_t$ with $\Lambda^{\tmop{Red}}_t$, where the latter is 
obtained by replacing Eq.~\eqref{eq:nz} with
\begin{equation}\label{eq:red}
\frac{\mathd}{\mathd t} \Lambda^{\tmop{Red}}_t 
= \mathcal{K}^{\tmop{Red}}_t \Lambda_t, \quad  \mathcal{K}^{\tmop{Red}}_t
= \int_0^t \mathd \tau \mathcal{K}^{\tmop{NZ}}_{\tau}.
\end{equation}
When this approximation is used in the presence of a weak-coupling interaction between the 
open system and its environment, the resulting equation
is often called Redfield equation \cite{Redfield1957a,Breuer2002}.
We will refer to Eq.~\eqref{eq:red} as
Redfield-like master equation, in order to emphasize that we take it into account without necessarily restricting to the weak-coupling regime. 
The Redfield equation is commonly exploited in several different contexts, such as the study of
transport processes in condensed-matter or biophysical systems \cite{Kondov2001,Dassia2003,Timm2008,Jeske2015,Oviedo2016}.
Importantly, the Redfield equation might lead to a not well-defined evolution, as studied by now extensively in the literature \cite{Benatti2005,Whitney2008,Hartmann2020}. 
We will show how the damping-basis representation enables us to determine the structural properties of the resulting time-local generator
$\mathcal{K}^{\tmop{Red}}_t$, as well as to investigate the Markovian nature of the dynamics $\left\{\Lambda^{\tmop{Red}}_t\right\}_{t\geq0}$ .

\subsection{Structure of the time-local generator}\label{sec:Redfield1}

We first consider the following 
simple connection between the time-local generator, the memory kernel and the Redfield-like generator, 
in the case of commutative dynamics.
\begin{prop}\label{prop:red}
Under the same assumptions of Proposition \ref{prop:str}, for time-local generator and memory kernel in the form Eq.~\eqref{eq:ktcl} and \eqref{eq:knz} respectively,
the Redfield-like generator $\mathcal{K}^{\tmop{Red}}_t$ has the damping-basis diagonalization
\begin{eqnarray}
\mathcal{K}^{\tmop{Red}}_t &=& \sum_{\alpha=1}^{N^2} m^{\tmop{Red}}_{\alpha}(t) \mathcal{M}_{\alpha}, \label{eq:kred}
\end{eqnarray}
with
\begin{eqnarray}
m^{\tmop{Red}}_{\alpha}(t) &=& \int_0^t \mathd \tau m^{\tmop{NZ}}_{\alpha}(\tau) \label{eq:mred}.
\end{eqnarray}
Moreover, under the same assumptions of Corollary \ref{cor:str}, 
the Redfield-like generator $\mathcal{K}^{\tmop{Red}}_t $ reads
\begin{eqnarray}
\mathcal{K}^{\tmop{Red}}_t &=& \frac{m^{\tmop{Red}}(t)}{\ell}\mathcal{L}.\label{eq:corred}
\end{eqnarray}
\end{prop}
Accordingly, the Redfield-like dynamics has the same operational structure as the exact one, with time-dependent eigenvalues obtained from the corresponding original ones. This situation is depicted in Fig.~\ref{fig:prop:unique}. 
What is more, the eigenvalues $m^{\tmop{Red}}_{\alpha}(t)$
can be written in a compact form in terms of the functions $G_{\alpha}(t)$ defined in Eq.~(\ref{eq:ggammat}).
\begin{prop}\label{prop:coeff}
Under the same assumptions of Proposition \ref{prop:str}, consider a diagonal time-local generator as in Eq.~(\ref{eq:ktcl}).
The eigenvalues of the Redfield-like time-local generator satisfy the integral equation
\begin{equation}\label{eq:integro}
m^{\tmop{Red}}_{\alpha}(t) =G_{\alpha}(t) 
-\int_0^t \mathd \tau  G_{\alpha}(t-\tau) m^{\tmop{Red}}_{\alpha}(\tau),
\end{equation}
whose solution, provided
\begin{equation}
    |\widetilde{G_{\alpha}}(u)|<1, \label{eq:condpr}
\end{equation} 
can be written as
\begin{equation}\label{eq:recur}
m^{\tmop{Red}}_{\alpha}(t) =\sum_{j=1}^{\infty} (-)^{j+1} \underbrace{G_{\alpha} \ast\dots \ast G_{\alpha}}_{\tmop{j\, times}}(t).
\end{equation}
\end{prop}
\begin{proof} The proof is given in \ref{app:proofp3}.
\end{proof}

\subsection{Markovianity of the dynamics}\label{sec:Redfield2}
As mentioned in the Introduction, a highly relevant property of the dynamics of an open quantum system is its \mbox{(non-)Markovianity}. Quantum non-Markovianity is often defined in terms of the divisibility property, positive or completely positive, of the corresponding dynamical maps. 
As we will see below, the damping-basis representation enables us to make some statements about the preservation of (C)P-divisibility under 
the approximation leading to the Redfield-like master equation.

First, we note that the relation in Eq.~(\ref{eq:recur}) directly allows us to infer that (C)P-divisibility
is preserved in the Redfield-like master equation, whenever $\mathcal{L}$ has only one non-zero eigenvalue.
\begin{corollary}\label{cor:cpdiv}
Under the same assumptions of Proposition \ref{prop:coeff}, consider a time-local generator of the form $\mathcal{K}^{\tmop{TCL}}_t = \gamma(t) \mathcal{L}$,
such that $\mathcal{L}$ has only one eigenvalue $\ell$ different from zero, so that $\mathcal{L} = \ell \sum_{\alpha=1}^d \mathcal{M}_{\alpha}$,
and assume Eq.~(\ref{eq:condpr}) holds.
If $\left\{\Lambda_t\right\}_{t \geq0}$ is (C)P-divisible,
then $\left\{\Lambda^{\tmop{Red}}_t\right\}_{t \geq0}$ is (C)P-divisible.
\end{corollary}
\begin{proof}
First, we note that the time-local generators, $\mathcal{K}^{\tmop{TCL}}_t$
and $\mathcal{K}^{\tmop{Red}}_t$, are of the form 
$$
\mathcal{K}^{\tmop{TCL}}_t =\gamma(t) \mathcal{L}, \qquad 
\mathcal{K}^{\tmop{Red}}_t =\gamma^{\tmop{Red}}(t) \mathcal{L},
$$
where $\gamma^{\tmop{Red}}(t)=\frac{m^{\tmop{Red}}(t)}{l}$.
The dynamics $\left\{\Lambda_t\right\}_{t \geq0}$ 
($\left\{\Lambda^{\tmop{Red}}_t\right\}_{t \geq0}$) is (C)P-divisible if and only if
$\gamma(t)\geq 0$ ($\gamma^{\tmop{Red}}(t)\geq 0$).
Thus, since we assume $\left\{\Lambda_t\right\}_{t \geq0}$ (C)P-divisible, we
have $\gamma(t)\geq 0$. Moreover, $\ell$ has to be real and negative, since it is the only non-zero
eigenvalue of $\mathcal{L}$.
It follows that the function $G(t)$, see Eq.~(\ref{eq:mnz2}),
$$
G(t) = \ell \gamma(t) e^{\ell \Gamma(t)}
$$
is negative. But then, from Eq.~(\ref{eq:recur}) we have that $m^{\tmop{Red}}(t)/\ell$ is positive and thus  
$\left\{\Lambda^{\tmop{Red}}_t\right\}_{t \geq0}$ is (C)P-divisible.
\end{proof}
Note that the reverse statement, that is that  (C)P-divisibility of the Redfield-like dynamics originates from (C)P-divisibility of the original dynamics, is in general not true (even in the case of a single non-zero eigenvalue). Indeed, in the next example we will consider two dynamics, which both have a (C)P-divisible approximated evolution, but only one of them has this property originally.\\

\noindent{\textbf{Example 3}} \refstepcounter{dummy}\label{ex:3}
Let us consider the case of a pure dephasing
dynamics for a two-level system, characterized by a monotonic reduction of
coherences described by the decoherence funtion $\varphi (t)$. Starting from the expression of the dynamics we can
easily work out the corresponding time-local generator, which takes the form
\begin{equation}
  \mathcal{K}_t^{\tmop{TCL}} = \gamma (t) \mathcal{L},  \label{eq:3tcl}
\end{equation}
with
\begin{equation}
  \gamma (t) = - \frac{\dot{\varphi} (t)}{\varphi (t)} 
\end{equation}
and
\begin{equation}
  \mathcal{L} = \sigma_+ \sigma_- \cdot \sigma_+ \sigma_-  + \sigma_- \sigma_+ \cdot \sigma_- \sigma_+  - \mathds{1}. 
\end{equation}
It can be easily checked that $\mathcal{L}$
has only one non-zero two-fold degenerate eigenvalue, so that relying on Corollary \ref{cor:str} we get the corresponding memory kernel,
\begin{equation}
  \mathcal{K}_t^{\tmop{NZ}} = k (t) \mathcal{L},  \label{eq:3nz}
\end{equation}
with
\begin{equation}
  \tilde{k} (u)  =  - \frac{\widetilde{\dot{\varphi}} (u)}{\tilde{\varphi}
  (u)} . 
\end{equation}
Since Eq.~(\ref{eq:3nz}) exactly corresponds to Eq.~(\ref{eq:3tcl}), this
memory kernel describes the very same CPTP dynamics. In particular, it
describes a CP-divisible dynamics since $\varphi (t)$ is monotonically
decreasing. According to Proposition \ref{prop:red}, the Redfield-like
generator is given by
\begin{equation}
  \mathcal{K}_t^{\tmop{Red}} = \int^t_0 \mathd \tau  k (\tau)
  \mathcal{L},  \label{eq:3red}
\end{equation}
where the function $K (t) = \int^t_0 \mathd \tau  k (\tau)$ is indeed
positive since $\tilde{K} (u)= -{\widetilde{\dot{\varphi}} (u)}/({1 +
  \widetilde{\dot{\varphi}} (u)}) $
and $\varphi (t)$ is monotonically decreasing. The Redfield-like generator
therefore also describes a CP-divisible dynamics.
Let us now start from
Eq.~(\ref{eq:3nz}) changing the operatorial structure in the memory kernel via
the replacement $\mathcal{L} \rightarrow \overline{\mathcal{L}}$ with
\begin{equation}
  \overline{\mathcal{L}} = \mathcal{} \sigma_z \cdot \mathcal{} \sigma_z -
  \mathds{1}, 
\end{equation}
so that we define
\begin{equation}
  \overline{\mathcal{K}}_t^{\tmop{NZ}} = k (t) \overline{\mathcal{L}} . 
  \label{eq:3nzbar}
\end{equation}
The associated Redfield-like generator reads
\begin{equation}
  \overline{\mathcal{K}}_t^{\tmop{Red}} = \int^t_0 \mathd \tau  k
  (\tau) \overline{\mathcal{L}},  \label{eq:3redbar}
\end{equation}
and again corresponds to a CP-divisible dynamics. Making reference to
Corollary \ref{cor:str} (also $\overline{\mathcal{L}}$ has only one non-zero eigenvalue), we can obtain via the damping-basis approach the
time-local generator exactly corresponding to Eq.~(\ref{eq:3nzbar}), namely
\begin{equation}
  \overline{\mathcal{K}}_t^{\tmop{TCL}} =\bar{\gamma} (t)
  \overline{\mathcal{L}},  \label{eq:3tclbar}
\end{equation}
where now
\begin{equation}
\label{eq:3}
\bar{\gamma} (t) = - \frac{1}{2} \frac{1}{| \bar{\varphi} (t) |}
   \frac{\mathd}{\tmop{dt}} | \bar{\varphi} (t) |,
\end{equation}
 with
 \begin{equation}
   \label{eq:4}
   \bar{\varphi} (t)= \mathfrak{I}\left(\frac{1}{u}\frac{1+\tilde{\dot{\varphi}}(u)}{1-\tilde{\dot{\varphi}}(u)}\right)(t).
 \end{equation}
The sign of $\bar{\gamma} (t)$ is now not necessarily fixed, so that the dynamics is in general
not C(P)-divisible. The evolution is however always well defined, that is
CPTP, since $\int^t_0 \mathd \tau \bar{\gamma}  (\tau) \geqslant 0$
{\cite{Breuer2016a}}. The same therefore holds for the evolution described by
$\overline{\mathcal{K}}_t^{\tmop{NZ}}$, due to the exact correspondence
between Eq.~(\ref{eq:3tclbar}) and Eq.~(\ref{eq:3nzbar}). We have therefore shown an example of
two dynamics whose associated Redfield-like master equations both describe a
CP-divisible collection of maps, despite the fact that only one of the two was
originally CP-divisible, while the other can even break P-divisibility.\\
If we go beyond the case of only one non-zero eigenvalue, the situation gets more involved.
Focusing in particular on Pauli maps, one can see that P-divisibility is more robust under
the Redfield approximation than CP-divisibility.
\begin{corollary}\label{cor:pdivPauli}
  Consider a Pauli dynamical map in the form
\begin{equation}\label{eq:pauliTCL}
\frac{\mathd}{\tmop{dt}} \rho ( t )= \frac{1}{2}\sum\limits_{k=1}^{3} \gamma_k(t)(\sigma_k\rho_t\sigma_k - \rho_t).
\end{equation}
Assume that the dynamical map is P-divisibile and Eq.~\eqref{eq:condpr} holds. Then the associated Redfield-like map is also P-divisible.
\end{corollary}
\begin{proof}
For Pauli dynamical maps the eigenvalues of the time-local generator read: $m^{TCL}_k(t)=-(\gamma_i(t)+\gamma_j(t))$, $k \neq i \neq j$. A map in such a form is P-divisible iff
\cite{chruscinski,chruscinski2}
\begin{equation}\label{eq:gg}
  \gamma_i(t) + \gamma_j(t) \geq 0 \ , \ \ i\neq j  \Leftrightarrow m^{TCL}_\alpha(t)\leq 0  \ , \forall \alpha .
\end{equation}
Accordingly, $G_\alpha(t) \leq 0 , \forall \alpha$. Using this result, from Eq.~\eqref{eq:recur} we conclude that $m^{\tmop{Red}}_\alpha(t)\leq 0, \forall \alpha$, and the Redfield-like map is then also P-divisible. 
\end{proof}
On the other hand, CP-divisibility can be lost in the approximation leading to the Redfield-like master equation, as shown in
the following example; indeed, this is a consequence of the presence of more than one non-zero eigenvalues.\\

\noindent {\textbf{Example 4}}\refstepcounter{dummy} \label{ex:4}
We analyze the dephasing dynamics in random directions first introduced in \cite{moj-nonMarkovianity}:
\begin{equation}
  \label{eq:1}
  \Lambda_t = x_1 e^{\mathcal{L}_1 t}+x_2 e^{\mathcal{L}_2 t}+x_3 e^{\mathcal{L}_3 t},
\end{equation}
where $\mathcal{L}_i[\rho]=\sigma_i\rho \sigma_i-\rho$, describing dephasing along the $i$-th direction, and the $x_i$'s are the corresponding probabilities constrained by $x_1+x_2+x_3 = 1$.
The time evolution obeys the following time-local master equation
\begin{equation}
\frac{\mathd}{\tmop{dt}} \rho ( t )=\frac{1}{2}\sum\limits_{k=1}^{3} \gamma_k(t)(\sigma_k\rho(t)\sigma_k - \rho(t)) \label{ru-meq},
\end{equation}
where the time dependent decoherence rates can be expressed as
\begin{eqnarray}
  \gamma_1(t)  &=& \mu_1(t) - \mu_2(t) - \mu_3(t) \nonumber \\
  \gamma_2(t) &=&  -\mu_1(t) + \mu_2(t) - \mu_3(t) \label{gammas3} \\
   \gamma_3(t) &=& -\mu_1(t) - \mu_2(t) + \mu_3(t), \nonumber
\end{eqnarray}
with
\begin{align}\label{}
  \mu_1(t) = - \frac{x_2+x_3}{x_2+x_3 + e^{2t} x_1} &&
  \mu_2(t) = -  \frac{x_3 +x_1}{x_3+x_1 + e^{2t} x_2} \nonumber \\
	\mu_3(t) = - \frac{x_1+x_2}{x_1+x_2 + e^{2t} x_3}. \nonumber
\end{align}
The dynamics given by Eq.~\eqref{ru-meq} is always P-divisible, since Eqs.~\eqref{eq:gg} are satisfied. However, depending on the choice of the $x_i$ parameters one of the decay rates can become negative so that the evolution looses its CP-divisibility, see Fig.~\ref{fig:LoosingCPdiv} \textbf{(a)}. In particular, the well-known eternally non-Markovian master equation belongs to this family  \cite{hallcresserandersson}. It corresponds to the choice $x_1=x_2=\frac{1}{2}$, $x_3=0$ and generates an evolution which for all $t>0$ is non-CP-divisible, as in this case $\gamma_3(t)=-\tanh(t)<0, \forall t>0$.
\begin{figure}[!tbp]
  \centering
  \begin{minipage}[b]{6cm}
    \includegraphics[width=5.8cm]{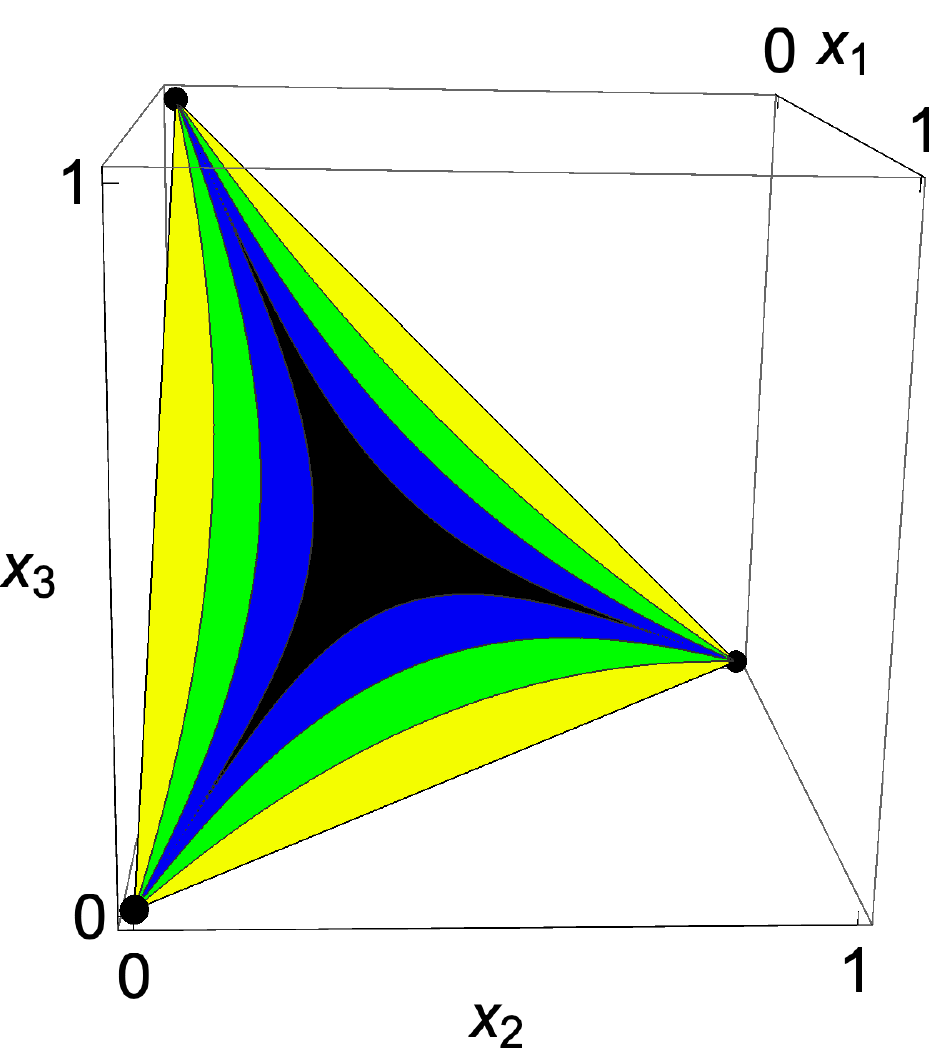}
    \hspace*{1.2cm}\textbf{a)}\textbf{ Exact dynamics}\\
  \end{minipage}
  \hfill
  \begin{minipage}[b]{6cm}
    \includegraphics[width=5.8cm]{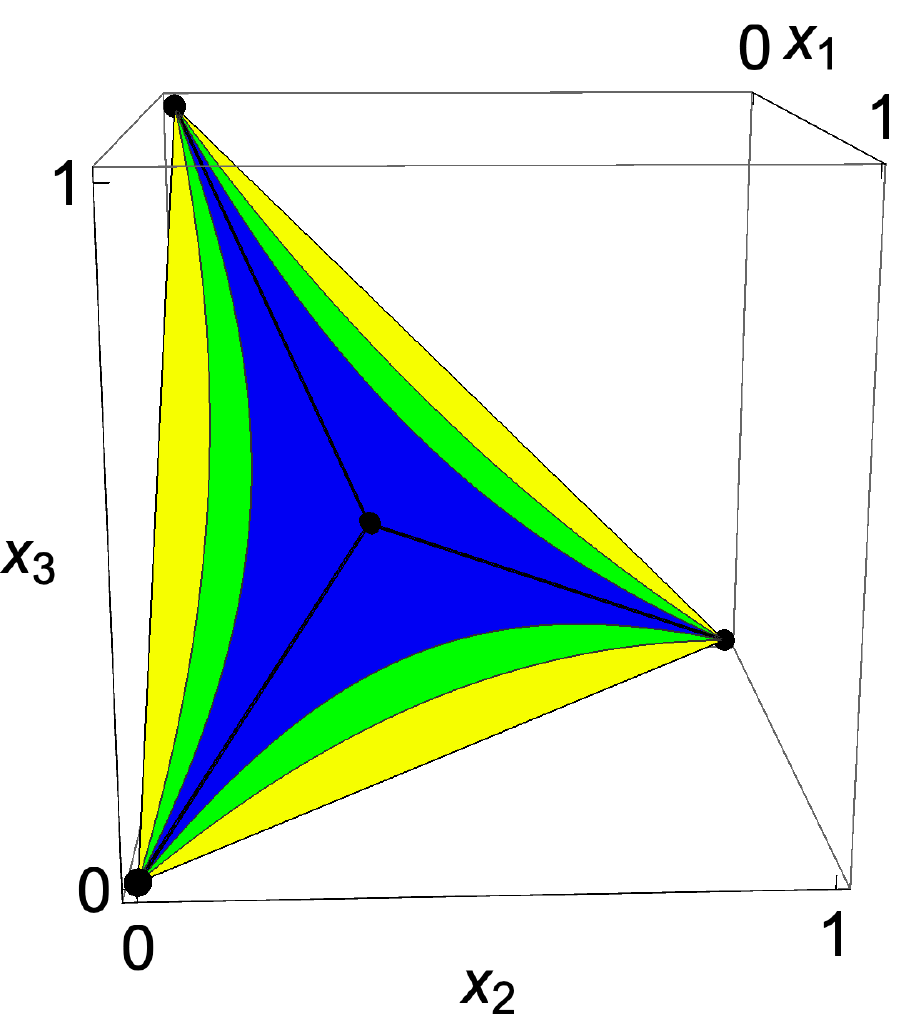}
    \hspace*{0.1cm} \textbf{b)}\textbf{ Approximated dynamics}\\
  \end{minipage}
    \hfill
  \begin{minipage}[b]{1cm}
    \includegraphics[width=1cm]{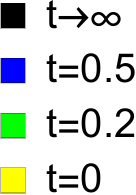}
    \vspace*{3.cm}
  \end{minipage}
  \caption{The areas depict the range of parameters $x_1,x_2,x_3$, for which all of the $\gamma_k(t)$ (panel \textbf{a)} the exact dynamics) and $\gamma^{\tmop{Red}}_k(t)$ (panel \textbf{b)} the approximated dynamics) are positive, that is the dynamics is CP-divisible, up to the given time; different colours of the areas correspond to different times (in arbitrary units). Initially ($t=0$, yellow triangle) the only constraint is in both cases the condition $x_1+x_2+x_3=1$.
    For long enough times $\gamma^{\tmop{Red}}_k(t)>0$ implies $\gamma_k(t)>0$
    . The black color corresponds to $t\rightarrow \infty$, so that for these choices of $x_1,x_2,x_3$ the corresponding dynamical map always remains CP-divisible. Only the choices $x_i=x_j\leq x_k$, with $i \neq j \neq k$ lead to CP-divisible Redfield-like dynamics (panel \textbf{b)} tripod configuration consisting of three lines), in contrast to the exact dynamics (panel \textbf{a)} black star shape). The four black dots denote choices of parameters $x_1, x_2, x_3$ for which both the original and the approximated generator have only one non-zero eigenvalue and the corresponding dynamical map is CP-divisible (in panel \textbf{a)} the dot in the middle is not visible).
  }\label{fig:LoosingCPdiv}
\end{figure}
The evolution resulting from the Redfield-like master equation Eq.~\eqref{eq:red} has 
the form
\begin{equation}
\frac{\mathd}{\tmop{dt}} \rho ( t )=\frac{1}{2}\sum\limits_{k=1}^{3} \gamma^{\tmop{Red}}_k(t)(\sigma_k\rho(t)\sigma_k - \rho(t)) \label{approx-eq},
\end{equation}
with
\begin{align}
\gamma^{\tmop{Red}}_1(t)&=Y_1(t)-Y_2(t)-Y_3(t),\\
\gamma^{\tmop{Red}}_2(t)&=-Y_1(t)+Y_2(t)-Y_3(t),\\
\gamma^{\tmop{Red}}_3(t)&=-Y_1(t)-Y_2(t)+Y_3(t),
\end{align}
and $Y_k(t)=\exp(-2x_k t)(x_k-1)$. This is still a proper CPTP quantum dynamics, as it fulfils the relevant constraints given in \cite{chruscinski}.
The dynamics \eqref{approx-eq} is P-divisible iff the condition \eqref{eq:gg} for $\gamma^{\tmop{Red}}_k(t)$ is satisfied. As
\begin{align}
\gamma^{\tmop{Red}}_i(t)+\gamma^{\tmop{Red}}_j(t)= 2(1-x_k)e^{-2x_k t}, \qquad i \neq j \neq k,
\end{align}
is always positive, by conducting our approximation the dynamics keeps its P-divisibility (in accordance with Corollary \ref{cor:pdivPauli}).
However, almost all of the dynamics which were CP-divisible for the original map, are no more CP-divisible for the Redfield-like master equation, see Fig.~\ref{fig:LoosingCPdiv}\textbf{ (a)}, for the original dynamics and Fig.~\ref{fig:LoosingCPdiv}\textbf{ (b)}, for the approximated evolution. Initially, for $t=0$, all $\gamma(t)$'s and $\gamma^{\tmop{Red}}(t)$'s are positive. 
After a transient behavior, positivity of $\gamma^{\tmop{Red}}(t)$ implies positivity of $\gamma(t)$, since the parameter range warranting positivity of the time dependent coefficients in the Redfield-like master equation is smaller.
The Redfield-like dynamical maps are CP-divisible only for the choices $x_i=x_j\leq x_k$, with $i \neq j \neq k$ (all $\gamma^{\tmop{Red}}(t)$'s are positive for all times). This is in contrast to the original evolution, where a significant fraction of the dynamics are CP-divisible \cite{moj-nonMarkovianity}. Note that there are only four choices of parameters $x_1, x_2, x_3$ (black dots in Fig.~\ref{fig:LoosingCPdiv}), for which the resulting original generator has only one eigenvalue and the corresponding dynamical map is CP-divisible. In accordance with Corollary \ref{cor:str} and Corollary \ref{cor:cpdiv} these properties are also present in the approximated dynamics.\\
These results put into question the idea that the Redfield-like approximation represents
a coarse graining in time of the open-system dynamics leading toward a memoryless evolution. Indeed, our analysis is complementary with respect to the investigation of the validity of the approximation itself given
the overall system-environment microscopic description \cite{Benatti2005,Whitney2008,Hartmann2020}.

The results of this Section 
are summarized in Table~\ref{tab:CP-div}.
When going to higher dimensional systems the situation gets more involved and no easy connections between the (C)P-divisibility of the exact and approximated dynamics were found. However, in a case of the generalized Pauli channels \cite{KatarzynaSiud2016} some statements are still possible, as elaborated in \ref{app:pdivGen}.

\begin{table}[h!]
\centering
 \renewcommand{\arraystretch}{1.3}
\begin{tabularx}{1.0\textwidth} { 
  | >{\raggedright\arraybackslash}X 
  !{\vrule width 2pt} >{\centering\arraybackslash}X 
  >{\centering\arraybackslash}X
   >{\centering\arraybackslash}X | }
 \hline
  & exact dynamics & & Redfield-like
\\
 \noalign{\hrule height 2pt}
 single eigenvalue  & C(P)-div &  $\Rightarrow$ & \normalsize C(P)-div   \\
& &$\nLeftarrow$ &\\
  \hline
 Pauli channel  
  & P-div &$\Rightarrow$ & P-div  \\
  & CP-div &$\nRightarrow$  & CP-div  \\
   \hline
\end{tabularx}
\caption{Connections between C(P)-divisibility property for exact and Redfield-like approximated dynamics.
Arrows $\Rightarrow$ and barred arrows $\nRightarrow$ express the implications shown by means of proofs and counterexamples (in the case of a single
non-zero eigenvalue by Corollary \ref{cor:cpdiv} and Example \hyperref[ex:3]{3}, and in the case of a Pauli channel by Corollary \ref{cor:pdivPauli} and Example \hyperref[ex:4]{4}, respectively).}
\label{tab:CP-div}
\end{table}

\section{Summary}\label{sec:summ}

In this paper, we investigated the connection between the time-local and the memory-kernel master equations associated with 
the dynamics of an open quantum system. Focusing on the class of commutative dynamics and making use of the damping-basis
approach, we formulated a general strategy to obtain each kind of master equation from the other one.
Only transformations among functions of time and their Laplace transforms are involved,
as the operational structure of the two master equations coincide in the damping-basis picture. 
In the presence of a single non-zero eigenvalue also the Lindbladian structure
of the time-local generator and the memory kernel is exactly the same.
Instead, when more eigenvalues are present, new Lindbladian operators can be generated in going from the time-local
master equation to the integro-differential one, or viceversa.

Furthermore, we analyzed the impact of the approximation leading to the Redfield-like master equation on both
the operatorial structure of the master equation and the Markovianity of the dynamics.
In the case of a single non-zero eigenvalue both CP-divisibility and P-divisibility are preserved,
but this is no longer guaranteed for more general dynamics.
In particular, restricting to Pauli dynamical maps, we showed that, while P-divisibility is still preserved,
CP-divisibility can be lost as a result of the Redfield-like coarse graining in time. 

Our results highlight the relevance of describing open system dynamics with different representations -- such as the time-local or the integro-differential master equations, as well as the damping-basis or the Lindbladian picture -- and the possibility to find direct connections among them. Indeed, it will be of interest to push forward such an analysis and deal with more general types of evolution in order to shed light on their relevance and usefulness for the description and characterization of non-Markovianity.

\ack

NM acknowledges funding by the Alexander von Humboldt Foundation in form of a Feodor-Lynen Fellowship, AS and BV acknowledge support from the FFABR project of MIUR. BV acknowledges support from the Joint Project ``Quantum Information Processing in Non-Markovian Quantum Complex Systems'' funded by FRIAS, University of Freiburg and IAR, Nagoya University. All authors acknowledge support from the Unimi Transition Grant H2020.

\section*{References}

\bibliographystyle{unsrt}
\bibliography{tclMkRed} 

\appendix

\section{Proof of Proposition \ref{prop:str}}\label{app:proofp1}
\begin{proof}
First, we recall that, assuming Eq.~(\ref{eq:ktcl}), Eq.~(\ref{eq:tcl}) 
is solved by (see Eq.~(\ref{eq:mapdiag}))
\begin{equation}
\Lambda_t  = \sum_{\alpha=1}^{N^2}  m_{\alpha}(t) \mathcal{M}_{\alpha},
\end{equation}
where $m_{\alpha}(t)$ is defined
via
\begin{equation}\label{eq:mmat}
   \log  m_{\alpha}(t) = \int_0^t \mathd \tau m^{\tmop{TCL}}_{\alpha}(\tau)
\end{equation}
and we have used
\begin{equation}
\sum\limits_{\alpha=1}^{N^2}\mathcal{M}_{\alpha}=\id, \qquad \mathcal{M}_{\alpha}\mathcal{M}_{\beta}=\delta_{\alpha,\beta}\mathcal{M}_{\alpha}.
\nonumber
\end{equation}
Moving to the Laplace domain due to linearity we get
\begin{equation}\label{eq:a1}
\widetilde{\Lambda}_u = \sum_{\alpha=1}^{N^2} \widetilde{m_{\alpha}}(u) \mathcal{M}_{\alpha};
\end{equation}
and for the memory kernel
\begin{equation}\label{eq:a2}
\widetilde{\mathcal{K}^{\tmop{NZ}}}_u = u \mathds{1} - \left(\widetilde{\Lambda}_u\right)^{-1}.
\end{equation}
Eq.~(\ref{eq:a1}) implies
\begin{equation}\label{eq:a3}
\left(\widetilde{\Lambda}_u\right)^{-1} = \sum_{\alpha=1}^{N^2} 
\left(\widetilde{m_{\alpha}}(u) \right)^{-1} \mathcal{M}_{\alpha},
\end{equation}
which, replaced in Eq.~(\ref{eq:a2}), 
yields Eq.~(\ref{eq:knz}) 
with (using also $\sum_{\alpha=1}^{N^2}\mathcal{M}_{\alpha} = \mathds{1}$)
$$
m^{\tmop{NZ}}_{\alpha}(t) = \mathfrak{I}\left(u-\left(\widetilde{m_{\alpha}}(u) \right)^{-1}\right);
$$
Eq.~(\ref{eq:mnz}) follows from the property of the Laplace transform of a derivative
$\tilde{\dot{f}}(u) = u \widetilde{f}(u)-f(0)$. We further have, from the definition of $G_{\alpha}(t)$ in Eqs.~(\ref{eq:ggammat})
\begin{equation}
  \label{eq:5}
  m_{\alpha}(t) =1+ \int_0^t \mathd \tau G_{\alpha}(\tau)
\end{equation}
so that
\begin{equation}
  \label{eq:6}
  m^{\tmop{TCL}}_{\alpha}(t)=\frac{\dot m_{\alpha}(t)}{m_{\alpha}(t)}
\end{equation}
implying Eq.~(\ref{eq:tcltcl}).
\end{proof}

\section{Proof of Proposition \ref{prop:coeff}}\label{app:proofp3}
\begin{proof}
Using Eqs.~(\ref{eq:mnz}) and (\ref{eq:mred}) and going to the Laplace domain, 
we have 
$$
\widetilde{m^{\tmop{Red}}_{\alpha}}(u) = \frac{1}{u} \widetilde{m^{\tmop{NZ}}_{\alpha}}(u)
= \frac{\widetilde{G_{\alpha}}(u)}{1+\widetilde{G_{\alpha}}(u)},
$$
from which
$$
\widetilde{m^{\tmop{Red}}_{\alpha}}(u)= \widetilde{G_{\alpha}}(u)-\widetilde{G_{\alpha}}(u)\widetilde{m^{\tmop{Red}}_{\alpha}}(u),
$$
which directly implies Eq.~(\ref{eq:integro}) when going back to the time domain.
Eq.~(\ref{eq:recur}) is obtained from Eq.~(\ref{eq:integro}) by iteration
or, equivalently, via the geometric series, whose convergence is guaranteed by Eq.~(\ref{eq:condpr}).
\end{proof}
\section{Lindbladian form: higher dimensional example}\label{app:proofhigherdim}
We will consider now an example in higher dimensions, $\mathcal{H}=\mathbb{C}^3$, showing how the
dimensionality strongly enhances the lack of correspondence
in the operatorial structures of the time-local generator and memory kernel.

Consider the (normalized) Gell-Mann matrices
\begin{gather*}
\sigma_1 =\frac{1}{\sqrt{2}} \matthree {0}{1}{0}{1}{0}{0}{0}{0}{0},\quad
\sigma_2 = \frac{1}{\sqrt{2}} \matthree {0}{-i}{0}{i}{0}{0}{0}{0}{0},\quad
\sigma_3 = \frac{1}{\sqrt{2}} \matthree {1}{0}{0}{0}{-1}{0}{0}{0}{0},\\[1ex]
\sigma_4 = \frac{1}{\sqrt{2}} \matthree {0}{0}{1}{0}{0}{0}{1}{0}{0},\quad
\sigma_5 = \frac{1}{\sqrt{2}} \matthree {0}{0}{-i}{0}{0}{0}{i}{0}{0},\quad
\sigma_6 = \frac{1}{\sqrt{2}} \matthree {0}{0}{0}{0}{0}{1}{0}{1}{0},\\[1ex]
\sigma_7 = \frac{1}{\sqrt{2}} \matthree {0}{0}{0}{0}{0}{-i}{0}{i}{0},\quad
\sigma_8 = \frac{1}{\sqrt{6}} \matthree {1}{0}{0}{0}{1}{0}{0}{0}{-2}
\end{gather*}
along with 
\begin{gather*}
S_+ = \matthree {0}{1}{0}{0}{0}{1}{0}{0}{0},\quad
S_- = \matthree {0}{0}{0}{1}{0}{0}{0}{1}{0}.\quad
\end{gather*}
The latter are the ladder operators in $\mathbb{C}^3$, and then the Example 2
is generalized to 3-level systems by looking at the memory kernel
[compare with Eq.~(\ref{eq:exnz2})]
\begin{equation}\label{eq:ex4}
   \mathcal{K}^{\tmop{NZ}}_t(\omega) = k(t) \left(S_- \omega S_+ -\frac{1}{2}(S_+ S_- \omega + \omega S_+ S_- ) 
   + S_+ \omega S_- -\frac{1}{2}(S_- S_+ \omega + \omega S_- S_+ )
    \right).
\end{equation}
This is still a self-adjoint generator, with now four non-zero eigenvalues,
and the corresponding damping-basis decomposition reads
\begin{eqnarray}
\hspace{-2,5cm}
\left\{m^{\tmop{NZ}}_\alpha(t)\right\}_{\alpha} = \left\{0,-3k(t),-\frac{5}{2}k(t),-\frac{5}{2}k(t),-k(t),-k(t),-k(t), -\frac{1}{2}k(t),-\frac{1}{2}k(t)\right\}\label{eq:valueex4}\\
\hspace{-2,5cm}
\left\{\tau_\alpha\right\}_{\alpha} = \left\{\frac{\mathds{1}}{\sqrt{3}},
-\frac{\sqrt{3}\sigma_3}{2}+\frac{\sigma_8}{2},
-\frac{\sigma_2}{\sqrt{2}}+\frac{\sigma_7}{\sqrt{2}},
-\frac{\sigma_1}{\sqrt{2}}+\frac{\sigma_6}{\sqrt{2}},
\frac{\sigma_3}{2}+\frac{\sqrt{3}\sigma_8}{2},
\sigma_4,
\sigma_5,\right. \nonumber \\ \left. \hspace{7,5cm}
,\frac{\sigma_2}{\sqrt{2}}+\frac{\sigma_7}{\sqrt{2}},
\frac{\sigma_1}{\sqrt{2}}+\frac{\sigma_6}{\sqrt{2}}
\right\}. \label{eq:vecex4}
\end{eqnarray}
Once again, applying Proposition 
\ref{prop:str} both the time local generator and the memory kernel
can be written as
\begin{equation}
  \mathcal{K}^X_t(\omega) = \sum_{\alpha=0}^8
    m^X_\alpha(t) \mbox{Tr}\left[\tau_\alpha \omega\right] \tau_\alpha   \qquad X = \tmop{TCL}, \tmop{NZ},\label{eq:ex4aux}
\end{equation}
where the memory-kernel eigenvalues are given by Eq.~(\ref{eq:valueex4}),
while the time-local-generator ones are obtained via Eqs.(\ref{eq:mtcl}) and (\ref{eq:ggammat});
the eigenbasis is indeed given in both cases by Eq.~(\ref{eq:vecex4}).
By means of Eqs.~(\ref{eq:coeffl})-(\ref{eq:diagop}), the generator above can be written in
Lindbladian form as
\begin{eqnarray}
    \mathcal{K}^X_t &=& \frac{1}{6}(m^X_1(t)-3m^X_3(t)) \left(
    \mathcal{D}_{(\sigma_3+\sqrt{3}\sigma_8)/2} +\mathcal{D}_{\sigma_4}+\mathcal{D}_{\sigma_5}
    \right)\nonumber\\
    &&+\frac{1}{6}(3m^X_1(t)-4m^X_2(t)+3m^X_3(t)-4m^X_4(t)) \mathcal{D}_{(-\sqrt{3}\sigma_3+\sigma_8)/2} \nonumber\\
    &&-\frac{1}{6}(2m^X_1(t)-3m^X_2(t)+3m^X_4(t))
    \left(
    \mathcal{D}_{(-\sigma_2+\sigma_7)/\sqrt{2}}+\mathcal{D}_{(-\sigma_1+\sigma_6)/\sqrt{2}}\right) \nonumber\\
    &&-\frac{1}{12}(2m^X_1(t)+3m^X_2(t)-3m^X_4(t)
    \left(\mathcal{D}_{S_+}+\mathcal{D}_{S_-}\right),
\end{eqnarray}
where we used $\mathcal{D}_A$ to denote the Lindblad dissipator with respect to the operator $A$.
The previous expression reduces to Eq.~(\ref{eq:ex4}) for eigenvalues as in Eq.~(\ref{eq:valueex4}),
so that one recovers the Lindblad-operator structure of the memory kernel with two contributions only, while the time-local
generator will have in general 8 Lindblad operators.

\section{P-divisibility in higher dimensions}\label{app:pdivGen}

For generalised Pauli dynamical maps, no sufficient and necessary conditions for P-divisibility are known. However, with the necessary condition as given in \cite{KatarzynaSiud2016}, one can obtain the following corollary.

\begin{corollary}\label{cor:pdivGenPauli}
Consider a generalized Pauli dynamical map for which Eq.~\eqref{eq:condpr} holds. If at least one eigenvalue of the Redfield-like time-local generator takes on positive values for some time $t\geq0$, then the original dynamical map is not P-divisible.
\end{corollary}
\begin{proof}
Assume that the original dynamics is P-divisible. As elaborated in \cite{KatarzynaSiud2016}, from this it follows that $m^{TCL}_\alpha(t)\leq 0$
for any $t\geq0$. However from \eqref{eq:recur} we conclude then that $m^{\tmop{Red}}_\alpha(t)\leq 0  \ , \forall \alpha$. Accordingly, if for some $\alpha$ and $t\geq 0$ $m^{\tmop{Red}}_\alpha(t)>0$, then the original dynamical map is not P-divisible.
\end{proof}

\end{document}